\documentclass[11pt,a4paper,reqno]{amsart}%
\usepackage{amsthm,amsmath,amsfonts,amssymb,amsxtra,appendix,bookmark,euscript,delarray,dsfont}
\usepackage[latin1]{inputenc}
\usepackage{color}

\newtheorem{theorem}{Theorem}[section]
\newtheorem{proposition}[theorem]{Proposition}
\newtheorem{remark}[theorem]{Remark}
\newtheorem{lemma}[theorem]{Lemma}

\numberwithin{equation}{section}

\newcommand\R{{\ensuremath {\mathbb R} }}

\newcommand\1{{\ensuremath {\mathds 1} }}

\newcommand\nn{\nonumber}

\renewcommand\phi{\varphi}

\newcommand{\gH}{\mathfrak{H}}

\newcommand{\cM}{\mathcal{M}}

\newcommand{\cE}{\mathcal{E}}

\newcommand{\eps}{\epsilon}

\renewcommand{\epsilon}{\varepsilon}

\newcommand{\norm}[1]{ \left| \! \left| #1 \right| \! \right| }

\DeclareMathOperator{\Tr}{{\rm Tr}}

\renewcommand{\ge}{\geqslant}
\renewcommand{\le}{\leqslant}
\renewcommand{\geq}{\geqslant}
\renewcommand{\leq}{\leqslant}
\renewcommand{\hat}{\widehat}
\renewcommand{\tilde}{\widetilde}

\newcommand{\wh}{\widetilde h}
\newcommand{\wH}{\widetilde H}
\newcommand{\wPsi}{\widetilde \Psi}
\newcommand{\wmu}{\widetilde\mu}
\newcommand{\PL}{P_\Lambda}

\newcommand{\black}[1]{\textcolor{black}{#1}}

\allowdisplaybreaks

\title[The Gross-Pitaevskii limit revisited]{Ground states of large bosonic systems: The Gross-Pitaevskii limit revisited}

\author[P.~T. Nam]{Phan Th\`anh Nam}
\address{IST Austria, Am Campus 1, 3400 Klosterneuburg, Austria} 
\email{pnam@ist.ac.at}

\author[N. Rougerie]{Nicolas Rougerie}
\address{Universit\'e Grenoble 1 \& CNRS,  LPMMC (UMR 5493), B.P. 166, F-38 042 Grenoble, France}
\email{nicolas.rougerie@grenoble.cnrs.fr}

\author[R. Seiringer]{Robert Seiringer}
\address{IST Austria, Am Campus 1, 3400 Klosterneuburg, Austria} 
\email{robert.seiringer@ist.ac.at}

\begin{document}
\date{March 30th, 2015}

\begin{abstract}
We study the ground state of a dilute Bose gas in a scaling limit where the Gross-Pitaevskii functional emerges. This is a repulsive non-linear Schr\"odinger functional whose quartic term is proportional to the scattering length of the interparticle interaction potential. We propose a new derivation of this limit problem, with a method that bypasses some of the technical difficulties that previous derivations had to face. The new method is based on a combination of Dyson's lemma, the quantum de Finetti theorem and a second moment estimate for ground states of the effective Dyson Hamiltonian. It applies equally well to the case where magnetic fields or rotation are present.  
\end{abstract}

\maketitle

\setcounter{tocdepth}{2}
\tableofcontents

\section{Introduction}

The rigorous derivation of effective non-linear theories from many-body quantum mechanics has been studied extensively in recent years, motivated in part by experiments in cold atom physics. For bosons, the emergence of the limit theories can be interpreted as due to most of the particles occupying the same quantum state: this is the Bose-Einstein condensation phenomenon, observed first in dilute alkali vapors some twenty years ago. 

The parameter regime most relevant for the description of the actual physical setup is the Gross-Pitaevskii limit. It is also the most mathematically demanding regime considered in the literature so far, see~\cite{LieSeiYng-00,LieSei-02,LieSei-06} for the derivation of equilibrium states and~\cite{ErdYauSch-10,ErdSchYau-09,BenOliSch-12,Pickl-15} for dynamics (more extensive lists of references may be found in~\cite{LieSeiSolYng-05,rougerie-cdf,BenPorSch-15}). The main reason for this sophistication is the fact that interparticle correlations due to two-body scattering play a leading order role in this regime. The goal of this paper is to present a method for the derivation of Gross-Pitaevskii theory at the level of the ground state that is conceptually and technically simpler than existing proofs, in particular that of~\cite{LieSei-06} which was so far the only method applicable when an external magnetic field is present. 

Our setting is as follows: we consider $N$ interacting bosons in the three-dimensional space $\R ^3$, described by the many-body Schr\"odinger Hamiltonian
\begin{equation}\label{eq:start hamil}
 H_N =\sum_{j=1}^N h_j  +\sum_{1\le j<k \le N} w_N(x_j-x_k) 
\end{equation}
acting on the  space $\gH^N = \bigotimes_{\rm sym}^N L^2(\R^3)$ of permutation-symmetric square integrable functions. The one-body operator is given by 
$$ h:= (-i\nabla + A(x))^2+V(x)$$
with a magnetic (or a rotation) field  $A$ satisfying 
\begin{align} \label{eq:assumption-A}
A\in L^3_{\rm loc}(\R^3,\R^3),\quad \lim_{|x|\to \infty} |A(x)| e^{-b|x|}=0
\end{align}
for some constant $b>0$ and an external potential $V$ satisfying 
\begin{align} \label{eq:assumption-V}
0\le V\in L^1_{\rm loc}(\R^3),  \quad \lim_{|x|\to \infty} V(x)=+\infty. 
\end{align}
We thus consider non-relativistic particles in a trapping potential, possibly under the influence of an effective magnetic field, which might be due to rotation of the sample or the interaction with optical fields. 

The particles interact pairwise via a repulsive potential $w_N$ given by
\begin{align} \label{eq:assumption-w}
w_N(x)=N^2 w(Nx), 
\end{align} 
where $w$ is a fixed function which is non-negative, radial and of finite range, i.e., $\1(|x|>R_0)w(x)  \equiv 0$ for some constant $R_0>0$. Different scalings of the interaction potential have been considered in the literature, of the form
\begin{equation}\label{eq:wbeta} 
w_{\beta,N} (x) = \frac{1}{N} N ^{3\beta} w (N ^{\beta} x)
\end{equation}
with $0 \leq \beta \leq 1$. The $N ^{-1}$ prefactor makes the interaction energy in~\eqref{eq:start hamil} of the same order as the one-particle energy. Indeed, if $\beta>0$, then
\begin{equation}\label{eq:w delta}
 N ^{3\beta} w (N ^{\beta} x) \underset{N\to \infty}{\to} \left(\int w\right) \delta_0 
\end{equation}
weakly and thus the interaction potential $w_{\beta,N}$ should be thought of as leading to a bounded interaction energy per pair of particles. Generally speaking, the larger the parameter $\beta$, the faster the potential converges to a point interaction, and thus the harder the analysis. Note that the cases $\beta < 1/3$ and $\beta > 1/3$  correspond to two physically rather different  scenarios: in the former the range of the potential is much larger than the typical interparticle distance $N ^{-1/3}$ and we should expect many weak collisions; while in the latter we rather have very few but very strong collisions. In this paper we consider the most interesting case $\beta = 1$ where the naive approximation \eqref{eq:w delta} does {\em not} capture the leading order behavior of the physical system. In fact, the strong correlations at short distances $O(N^{-1})$ yield a nonlinear correction, which essentially amounts to replacing the coupling constant $\int w$ by $(8\pi) \times $ the scattering 
length of $w$.

Let us quickly recall the definition of the scattering length; a more complete discussion can be found in \cite[Appendix~C]{LieSeiSolYng-05}. Under our assumption on $w$, the zero-energy scattering equation 
$$ (-2\Delta+w(x))f(x), \quad \lim_{|x|\to \infty}f(x)=1,$$
has a unique solution and it satisfies
$$ f(x)=1-\frac{a}{|x|},~~\forall |x|>R_0$$
for some constant $a \ge 0$ which is called the {\em scattering length} of $w$. In particular, if $w$ is the potential for hard spheres, namely $w(x)\equiv \black{+\infty}$ when $|x|<R_0$ and $w(x)\equiv \black{0}$ when $|x|\ge R_0$, then the scattering length of $w$ is exactly $R_0$. In a dilute gas, the scattering length can be interpreted as an effective range of the interaction: a quantum particle far from the others is felt by them as a hard sphere of radius $a$. A useful variational characterization of $a$ is as follows:
\begin{equation}\label{eq:var scat}
8\pi a = \inf\left\lbrace \int_{\R ^3} 2 |\nabla f| ^2 + w |f| ^2, \quad \lim_{|x|\to \infty}f(x)=1 \right\rbrace. 
\end{equation}
Consequently, $8\pi a$ is smaller than $\int w$ (the strict inequality can be seen by taking the trial function $1-\lambda g$ with $g\in C^2_c(\R^3,\R)$ satisfying $g(x)\equiv 1$ when $|x|<R_0$, and $\lambda>0$ sufficiently small). Moreover, a simple scaling shows that the scattering length of $w_N=N^2 w(N.)$ is $a/N$. 

We are going to prove that the ground state energy and ground states of $H_N$ converge to those of the Gross-Pitaevskii functional 
\begin{equation}\label{eq:GP func}
 \cE_{\rm GP}(u):=\langle u, h u \rangle + 4\pi a \int_{\R^3} |u(x)|^4 dx 
\end{equation}
in a suitable sense. 
Note that the occurrence of the scattering length in \eqref{eq:GP func} is subtle: this functional is \emph{not} obtained by testing $H_N$ with factorized states of the form $u ^{\otimes N}$ (which would lead to a functional with $4\pi a$ replaced by $(1/2) \int w$). Taking into account the short-range correlation structure which gives rise to~\eqref{eq:GP func} is the main difficulty in the proof of the following theorem, which is our main result. 

\begin{theorem}[\textbf{Derivation of the Gross-Pitaevskii functional}]\label{thm:cv-GP}\mbox{}\\
Under conditions \eqref{eq:assumption-A}, \eqref{eq:assumption-V} and \eqref{eq:assumption-w}, we have
\begin{align} \label{eq:cv-energy}
\boxed{\lim_{N\to \infty} \inf_{\|\Psi\|_{\gH^N}=1} \frac{\langle \Psi, H_N \Psi\rangle}{N} = \inf_{\|u\|_{L^2(\R^3)}=1} \cE_{\rm GP}(u)=:e_{\rm GP}.} 
\end{align}
Moreover, if $\Psi_N$ is an approximate ground state for $H_N$, namely
$$
\lim_{N\to \infty}  \frac{\langle \Psi_N, H_N \Psi_N\rangle}{N} = e_{\rm GP},
$$
then there exists a subsequence $\Psi_{N_\ell}$ and a Borel probability measure $\mu$ supported on the set of minimizers of $\cE_{\rm GP}(u)$ such that 
\begin{align} \label{eq:cv-state}
\boxed{\lim_{\ell\to \infty}\Tr \left| \gamma_{\Psi_{N_\ell}}^{(k)} - \int |u^{\otimes k} \rangle \langle u^{\otimes k}| d\mu(u) \right| =0,\quad \forall k\in \mathbb{N}}
\end{align} 
where $\gamma_{\Psi_{N}}^{(k)}= \Tr_{k+1\to N} |\Psi_{N} \rangle \langle \Psi_{N}|$ is the $k$-particle  reduced density matrix of $\Psi_{N}$. In particular, if $\cE_{\rm GP}(u)$ subject to $\|u\|_{L^2}= 1$ has a unique minimizer $u_0$ has a unique minimizer $u_0$ (up to a complex phase), then there is complete Bose-Einstein condensation
\begin{align} \label{eq:BEC}
\boxed{\lim_{N\to \infty}\Tr \left| \gamma_{\Psi_{N}}^{(k)} - |u_0^{\otimes k} \rangle \langle u_0^{\otimes k}| \right| =0,\quad \forall k\in \mathbb{N}.}
\end{align}
\end{theorem}

The energy upper bound in \eqref{eq:cv-energy} was proved in \cite{LieSeiYng-00,Sei-03} (see also~\cite[Appendix~A]{BenPorSch-15} for an alternative approach). The energy lower bound in \eqref{eq:cv-energy} and the convergence of one-particle density matrices were proved in \cite{LieSei-06}. The simpler case $A\equiv 0$ had been treated before in~\cite{LieSeiYng-00} (ground state energy) and \cite{LieSei-02} (condensation). In this case, the uniqueness of the Gross-Pitaevskii minimizer $u_0$ follow from a simple convexity argument. The result in Theorem \ref{thm:cv-GP} is thus not new, but the existing proofs are fairly difficult, in particular that of~\cite{LieSei-06} which deals with the case $A\not \equiv 0$. 

In the present paper we will provide alternative proofs of the energy lower bound and the convergence of states using the quantum de Finetti theorem in the same spirit as in~\cite{LewNamRou-14,LewNamRou-14c}. Our proofs are conceptually and technically simpler than those provided in~\cite{LieSei-06}. The overall strategy will be explained in the next section.

Our result covers the case of a rotating gas where the minimizers of the GP functional can develop quantized vortices. This corresponds to taking $A(x)=\Omega \wedge x$ with $\Omega$ being the angular velocity vector. In this case, $V$ should be interpreted as the trapping potential minus $\frac{1}{2}(\Omega \wedge x)^2$. The assumption $V(x)\to \infty$ as $|x|\to \infty$ is to ensure that all particles are confined to the system. Here our conditions on $A$ and $V$ are slightly more general than those of \cite{LieSei-06} where $A$ is assumed to grow at most polynomially and $V$ is assumed to grow at least logarithmically.


The finite range assumption on $w$ is not a serious restriction because we can always restrict the support of $w$ to a finite ball without changing the scattering length significantly.  In fact, it is sufficient to assume that $w$ is integrable at infinity, in which case the scattering length is well-defined. We can also work with a more general interaction $w_N \ge 0$ (with scattering length $a_N$) rather than the specific choice   \eqref{eq:assumption-w}, as long as its range goes to zero and  $\lim_{N\to \infty} Na_N$ exists; then the result in Theorem \ref{thm:cv-GP} still holds with $a$ replaced by $\lim_{N\to \infty} Na_N$. In particular, if $w_N$ is chosen as in~\eqref{eq:w delta} for some $0<\beta<1$, then $Na_N \to (8\pi)^{-1}\int w$. The critical case $\beta=1$ considered in this paper is much more interesting because in the limit the true scattering length appears instead of its first order Born approximation $(8\pi)^{-1}\int w$.

\medskip

\noindent \textbf{Acknowledgment:} We thank Mathieu Lewin for helpful discussions. PTN is supported partially by the European Union's Seventh
Framework Programme under REA grant agreement
no. 291734. NR acknowledges financial support from the ANR (Mathostaq project, ANR-13-JS01-0005-01).

\section{Overall strategy} \label{sec:overall}

In this section we give an outline of the proofs of our main results, in order to better emphasize the key new points for the energy lower bound and the convergence of states.

We shall use the following notation: Let $\theta:\R^3 \to \R$ be a radial smooth Heaviside-like function, i.e. 
$$0\le \theta \le 1, \: \theta(x)\equiv 0 \mbox{ for } |x|\le 1 \mbox{ and } \theta(x)\equiv 1 \mbox{ for } |x|\ge 2.$$
Let $U:\R^3 \to \R$ be a radial smooth function supported on the annulus $1/2\le |x| \le 1$ such that 
$$U(x)\ge 0 \mbox{ and } \int_{\R ^3} U = 4\pi a.$$
For every $R >0$ define
$$ \theta_{R}(x)= \theta\Big(\frac{x}{R}\Big), \quad U_R(x) = \frac{1}{R^3} U\Big(\frac{x}{R}\Big).$$
The smooth cut-off function $\theta_R$ will be used to perform cut-offs in both space and momentum variables, the latter being always denoted by 
$$p=-i\nabla.$$
The potential $U_R$ will be used to replace the original one. The important points will be that the integral of $U_R$ yields the correct physical scattering length, and that we will have some freedom in choosing the range $R$ of $U_R$.

\medskip

\noindent {\bf Step 1 (Dyson's lemma).} The main difficulty in dealing with the GP limit is that an ansatz $u ^{\otimes N}$ does \emph{not} give the correct energy asymptotics. In this regime, correlations between particles \emph{do} matter, and one should rather think of an ansatz of the form 
\begin{equation}\label{eq:GP ansatz}
 \prod_{i=1} ^N u(x_i) \prod_{1\leq i< j \leq N} f (x_i-x_i), 
\end{equation}
or a close variant, where $f$ is linked to the two-body scattering process. 
We shall follow the approach of~\cite{LieSei-06}, relying on a generalization of an idea due to Dyson~\cite{Dyson-57}. The following lemma, proved in~\cite{LieSeiSol-05}, allows to bound our Hamiltonian from below by an effective one which is much less singular, but still encodes the scattering length of the original interaction potential. 

\begin{lemma}[\textbf{Generalized Dyson Lemma}]\label{lem:Dyson}\mbox{}\\
For all  $s>0$, $1>\eps>0$ and $R>2R_0/N$, we have 
\begin{align} \label{eq:Dyson-lemma}
H_N &\ge   \sum_{j=1}^N \Big( h_j - (1-\eps) p_j^2 \theta_s( p_j) \Big)  + \frac{(1-\eps)^2}{N} W_N - C\frac{N^2 R^2 s^5}{\eps}\,,
\end{align}
where
\begin{align}\label{eq:WN} 
\quad W_N := \sum_{i\ne j}^N U_R(x_i-x_j) \prod_{k\ne i,j} \theta_{2R} (x_j-x_k).
\end{align}
\end{lemma}

Here and in the sequel, $C$ stands for a generic positive constant.

\begin{proof} Recall that the scattering length of $w_N$ is $a/N$. Therefore, from \cite[Eq. (50) and the first estimate in (52), with $(v,a,\chi,s)$ replaced by $(w_N,a/N,\theta_s, s^{-1})$, respectively]{LieSeiSol-05} one has
$$ p^2 \theta_s(p) +\frac{1}{2} \sum_{j=1}^{N-1} w_N(x-y_j) \ge \frac{1-\eps}{N}\sum_{j=1}^{N-1} U_R(x-y_j) - \frac{CaR^2 s^5}{\eps} $$
on $L^2(\R^3)$, for all given points $y_j$ satisfying $\min_{j\ne k} |y_j-y_k| \ge 2R$. Since the left side is non-negative, we can relax the condition $\min_{j\ne k} |y_j-y_k| \ge 2R$ by multiplying the right side with $\prod_{k\ne j} \theta_{2R} (y_j-y_k)$. Thus for every $i=1,2,...,N$, 
\begin{align*} & p_i^2 \theta_s(p_i) +\frac{1}{2} \sum_{j\ne i}^{N} w_N(x_i-x_j) \\
&\ge \frac{1-\eps}{N}\sum_{j\ne i} U_R(x_i-x_j) \prod_{k\ne i,j} \theta_{2R}(x_j-x_k)- \frac{Ca R^2 {s^5}}{\eps }.
\end{align*}
Multiplying both sides with $1-\eps$ and summing over $i$ we obtain \eqref{eq:Dyson-lemma}.   
\end{proof}
\begin{proof}[Clarification] The reader should keep in mind that we will choose $R=R(N)\to 0$ (actually  $N^{-1/2}\gg R \gg N^{-2/3}$), then $s\to \infty$ and $\eps\to 0$. 
\end{proof} 

The main point of Dyson's lemma is that we can replace the hard interaction potential $w_N$ by a softer one $U_R$ which encodes the scattering length conveniently as $\int U_R =4\pi a$. The price we have to pay for this advantage is twofold, however: first, we have to use all the high-momentum part of the kinetic energy (note that $\theta_s( p)=1$ when $p\ge 2s$); and second, the new potential $U_R(x_i-x_j)$ comes with the cut-off $\prod_{k\ne i,j} \theta(x_j-x_k)$. Together they really describe a ``nearest neighbor" potential instead of an ordinary two-body potential. While the first problem is not too annoying as the low part of the momentum is sufficient to recover the full energy in the limit, the second problem is much more serious.

\medskip

\noindent{\bf Step 2 (Second moment  estimate).} The lower bound \eqref{eq:Dyson-lemma} leads us to consider the effective Hamiltonian 
\begin{equation}\label{eq:eff hamil}
\wH_N: = \sum_{j=1}^N \wh_j + \frac{(1-\eps)^2}{N} W_N 
\end{equation}
where
\begin{align}\label{eq:htilde}
\wh := h - (1-\eps) p^2 \theta_s(p) - \kappa_{\eps,s}, \quad \kappa_{\eps,s} := \inf \sigma \Big(h - (1-\eps) p^2 \theta_s(p) -1 \Big). 
\end{align}
Here we use the freedom to add and remove the constant $N\kappa_{\eps,s}$ to the Hamiltonian to reduce to the case $\wh\ge 1$. In order to ensure that $\kappa_{\eps,s}$ is finite, we need the  extra condition 
\begin{align}\label{eq:assumption-A-V}
\lim_{|x|\to \infty}\frac{|A(x)|^2}{V(x)}=0\,,
\end{align}
which can be removed at a later stage, as we shall explain below. 

We will now seek a lower bound to the ground state energy of~\eqref{eq:eff hamil}. The philosophy, as in the previous work~\cite{LieSei-06}, is that if $\Psi_N$ is the ground state of the original Hamiltonian, then roughly 
$$\Psi_N \approx \tilde{\Psi}_N \prod_{1\leq i< j \leq N} f (x_i-x_i)$$
where $f$ encodes the two-body scattering process and $\tilde{\Psi}_N$ is a ground state for~\eqref{eq:eff hamil}. Thus the Dyson lemma allows to extract the short-range correlation structure and we now want to justify that $\tilde{\Psi}_N$ can be approximated by a tensor power $u ^{\otimes N}$, that is, we want to justify the mean-field approximation at the level of the ground state of~\eqref{eq:eff hamil}.

There are two key difficulties left:
\begin{itemize}
\item  The effective Hamiltonian is genuinely many-body. It can be bound\-ed below by a three-body Hamiltonian, but obviously one will ultimately have to show that the three-body contribution can be neglected. 
\item To recover the correct energy in the limit we need to take $R\ll N ^{-1/3}$ in order to be able to neglect the three-body contribution in the effective Hamiltonian. We thus still have to deal with the mean-field approximation in the ``rare but strong collisions'' limit. In other words, even though the effective Hamiltonian is much less singular than the original one, we do not have the freedom to reduce the singularity as much as we would like. 
\end{itemize}

It is in treating these two difficulties that our new method significantly departs from the previous works~\cite{LewNamRou-14c,LieSei-06}. We shall rely on a  strong a priori estimate for ground states of~\eqref{eq:eff hamil}. In Lemma \ref{lem:a-priori-estimate}, we assume \eqref{eq:assumption-A-V} and show that (provided $R\gg N^{-2/3}$, which is sufficient for our purpose) 
\begin{equation} \label{eq:second-moment}
(\wH_N)^2 \ge \frac{1}{3}\Big(\sum_{j=1}^N \wh_j \Big)^2 .
\end{equation}
Note that a bound of this kind is not available for the original $H_N$ due to the singularity of its interaction potential. In particular, \eqref{eq:second-moment} implies that every ground state $\wPsi_N$ of $\wH_N$ satisfies the strong a-priori estimate 
\begin{align} \label{eq:intro-h1h2}
\langle \wPsi_N, \wh_1 \wh_2 \wPsi_N \rangle \le C_{\eps,s}.
\end{align}
This second moment  estimate is the key point in our analysis in the next steps. It is reminiscent of similar estimates used in the literature for the time-dependent problem~\cite{ErdSchYau-07,ErdSchYau-09,ErdYauSch-10,ErdYau-01}.

\medskip
\noindent\textbf{Notation.} We always denote by $C_{\eps}$ (or $C_{\eps,s}$) a (generic) constant independent of $s$, $N$ and $R$ (or independent of $N$ and $R$, respectively).

\medskip

\noindent{\bf Step 3 (Three-body estimate).} Next we have to remove the cut-off 
$$\prod_{k\ne i,j} \theta(x_j-x_k)$$
in $W_N$ to obtain a lower bound in terms of a two-body Hamiltonian. Using the elementary inequality (see \cite[Eq. (22)]{LieSei-06})
$$ \prod_{k:k\ne i,j} \theta_{2R}(x_j-x_k) \ge 1 - \sum_{k: k\ne i,j} (1-\theta_{2R}(x_j-x_k)) $$
we have
\begin{align} \label{eq:WN-UR} 
W_N \ge \sum_{i\ne j}^N U_R(x_i-x_j) - \sum_{k\ne i\ne j \ne k}U_R(x_i-x_j)(1-\theta_{2R}(x_j-x_k))
\end{align}
and we thus have only a three-body term to estimate. Since the summand in this term is zero except when $|x_i-x_j|\le R$ and $|x_j-x_k|\le 4R$, the last sum of \eqref{eq:WN-UR} can be removed if the probability of having three or more particles in a region of diameter $O(R)$ is small enough. This should be the case if $R$ is much smaller than $N^{-1/3}$, the average distance of particles, but it is rather difficult to confirm this intuition rigorously. 

In \cite{LieSei-06}, a three-body estimate was established using a subtle argument based on path integrals (the Trotter product formula). In this paper, we will follow a different, simpler approach. Instead of working directly with a ground state of $H_N$ as in \cite{LieSei-06}, we will consider a  ground state $\wPsi_N$ of the effective Hamiltonian $\wH_N$. Thanks to the second moment  estimate \eqref{eq:second-moment} we can show that (see Lemma \ref{lem:3-body})
\begin{align}\label{eq:intro-3-body} \sum_{k=3}^N \langle \wPsi_N, U_R(x_1-x_2)\theta_{2R}(x_2-x_k)\wPsi_N \rangle \le C_{\eps,s} NR^2
\end{align}
The right side of \eqref{eq:intro-3-body} is small with our choice $N^{-1/2}\gg R$. 

\medskip

\noindent {\bf Step 4 (Mean-field approximation).} With the cut-off in $W_N$ removed, $\wH_N$ turns into the two-body Hamiltonian  
$$
K_N := \sum_{j=1}^N \wh_j + \frac{(1-\eps)^2}{N} \sum_{i\ne j}U_R(x_i-x_j) 
$$
for which we can validate the {\em mean-field approximation}. This is the simplest approximation for Bose gases where one restricts the many-body wave functions to the pure tensor products $u^{\otimes N}$. Since $U_R$ converges to the delta-interaction with mass $\int U_R=4\pi a$, we formally obtain the following approximation for the ground state energy
$$e_{\rm NL}(\eps,s):=\inf_{\|u\|_{L^2}=1} \left( \langle u, \wh u \rangle +  (1-\eps)^2 4\pi a \int |u|^4 \right).
$$ 
In Section \ref{sec:MF} we will show that 
\begin{align} \label{eq:Hartree-cv}
\lim_{N\to \infty} \frac{\inf\sigma(K_N)}{N} = e_{\rm NL}(\eps,s) .
\end{align}

A similar result was proved in \cite{LieSei-06} using a coherent state method, which is a generalization of the c-number substitution in \cite{LieSeiYng-05}. In the present paper, we will provide an alternative proof of \eqref{eq:Hartree-cv} using the quantum de Finetti theorem of St\o rmer \cite{Stormer-69} and Hudson and Moody \cite{HudMoo-75}. We note that this theorem has proved useful also in the derivation of the GP equation in the dynamical case, see \cite{Chen}.
 The following formulation is taken from \cite[Corollary 2.4]{LewNamRou-14} (see~\cite{rougerie-cdf} for a general discussion and more references): 

\begin{theorem}[\textbf{Quantum de Finetti}]\label{thm:DeFinetti}\mbox{}\\
Let $\mathfrak{K}$ be an arbitrary separable Hilbert space and let $\Psi_N\in\bigotimes_{\rm sym}^N \mathfrak{K}$ with $\|\Psi_N\|=1$. Assume that the sequence of one-particle density matrices $\gamma^{(1)}_{\Psi_N}$ converges strongly in trace class when $N\to \infty$. Then, up to a subsequence, there exists a (unique) Borel probability measure $\mu$ on the unit  sphere $S\mathfrak{K}$, invariant under the group action of $S^1$, such that 
\begin{equation}\label{eq:cv-state-deF}
\lim_{N\to \infty} \Tr \left| \gamma^{(k)}_{\Psi_{N}}-\int |u^{\otimes k}\rangle\langle u^{\otimes k}| \, d\mu(u) \right|,\quad \forall k\in \mathbb{N}. 
\end{equation}
\end{theorem}

This theorem validates the mean-field approximation for a large class of trapped Bose gases, in particular (see \cite{LewNamRou-14} and references therein) when the strength of the interaction is proportional to the inverse of the particle number, case $\beta = 0$  in~\eqref{eq:w delta}. However, when the interaction becomes stronger, the mean-field approximation is harder to justify. The convergence \eqref{eq:Hartree-cv} with $R\gg N^{-2/15}$ was proved in \cite{LewNamRou-14c} by using a quantitative version of Theorem \ref{thm:DeFinetti} valid for finite dimensional spaces~\cite{ChrKonMitRen-07,Chiribella-11,LewNamRou-14b}. However, this range of $R$ is too small for our purpose because we are forced to choose $R\ll N^{-1/2}$ in the previous steps. 

In this paper, thanks to the strong a-priori estimate \eqref{eq:intro-h1h2} we are able to prove \eqref{eq:Hartree-cv} for the larger range $R \gg N^{-2/3}$. As in~\cite{LewNamRou-14c,LieSei-06} we localize the problem onto energy levels of the one-body Hamiltonian $\tilde{h}$ lying below a chosen cut-off $\Lambda$. At fixed $\Lambda$, it turns out that the projected Hamiltonian is bounded proportionally to $N$. We are thus in a usual mean-field scaling if we are allowed to take $N\to \infty$ first, and then $\Lambda \to \infty$ later. Taking limits in this order  demands a very strong control on the localization error made by projecting the Hamiltonian, however. This control is provided again by the moment estimate~\eqref{eq:intro-h1h2}.  

\medskip

Combining the arguments in steps 1-4, we can pass to the limit $N\to \infty$, then $s\to \infty$ and $\eps\to 0$ to obtain the energy  convergence \eqref{eq:cv-energy} under the extra condition \eqref{eq:assumption-A-V}. In Section \ref{sec:GP-energy} we remove this technical assumption using a concavity argument from \cite{LieSei-06} and a binding inequality which goes back to an idea in \cite{Lieb-84}.
\medskip
\text{}\\
{\bf Step 5 (Convergence of ground states).} In Section \ref{sec:GS} we prove the convergence of (approximate) ground states using the convergence of the ground state energy of a perturbed Hamiltonian and the Hellmann-Feynman principle. A similar approach was used in \cite{LieSei-06} to prove the convergence of the 1-particle density matrix. However, the quantum de Finetti theorem helps us to avoid the complicated convex analysis in \cite{LieSei-06}, simplifying the proof significantly and giving access to higher order density matrices. 

\section{Second moment estimate} \label{sec:2m}

In this section we consider the effective Hamiltonian obtained after applying the generalized Dyson lemma to the original problem, namely
$$
\wH_N = \sum_{j=1}^N \wh_j + \frac{(1-\eps)^2}{N} W_N\,,
$$
where $\wh$ and $W_N$ are defined in (\ref{eq:htilde}) and (\ref{eq:WN}), respectively.
We will work under the extra assumption \eqref{eq:assumption-A-V}. Since $A\in L_{\rm loc}^3(\R^3,\R^3)$ and $V$ grows faster than $|A|^2$ at infinity, for every $\eps>0$ we have 
$$(V/2 - 2\eps^{-1}|A|^2)_-\in L^{3/2}(\R^3)$$
and hence 
$$(\eps/4) p^2 +V/2 - 2\eps^{-1}|A|^2 \ge -C_\eps.$$
In combination with the Cauchy-Schwarz inequality, we get
\begin{align*}
h- (1-\eps)p^2\theta_s(p) &\geq \frac{\eps}{2} p ^2 - 2\eps^{-1} |A|^2 + V \geq  \frac{\eps}{4}p^2+ \frac{V}{2}-C_{\eps}. 
\end{align*}
Therefore, $\inf\sigma(h)-1\ge \kappa_{\eps,s} \ge -C_{\eps}$ and
\begin{align} \label{eq:wh-lower}
\wh \geq C_{\eps}\left( -\Delta+V+1 \right).
\end{align}

The key estimate in this section is the following

\begin{lemma}[\textbf{Second moment estimate}]\label{lem:a-priori-estimate}\mbox{}\\
Assume that \eqref{eq:assumption-A-V} holds. For every $1>\eps>0$ and $s>0$, if  
$$R=R(N) \gg N^{-2/3}$$
when $N\to \infty$, then for $N$ large enough we have the operator bound
\begin{equation} \label{eq:h1h2}
(\wH_N)^2 \ge \frac{1}{3} \Big(\sum_{j=1}^N \wh_j \Big)^2.
\end{equation} 
\end{lemma}
We will show in Subsection~\ref{sec:3-body} that a convenient lower bound to Dyson's potential $W_N$ in terms of truly two-body operators follows from Lemma~\ref{lem:a-priori-estimate}.

Before proving Lemma \ref{lem:a-priori-estimate} in Subsection~\ref{sec:proof a priori}, we first collect some useful inequalities on a generic translation-invariant interaction operator $W(x-y)$ that will be used throughout the paper. 

\subsection{Operator inequalities for interaction potentials}\label{sec: op ineq}

We state several useful inequalities in the following lemma. In fact~\eqref{eq:W1} is well-known and~\eqref{eq:W2} with $\delta=0$ was proved earlier in \cite[Lemma 5.3]{ErdYau-01}. In the sequel we will crucially rely on the improvement to $\delta >0$, and on~\eqref{eq:W3} which seem to be new.

\begin{lemma}[\textbf{Inequalities for a repulsive interaction potential}]\label{lem:W}\mbox{}\\
For every $0\le W\in L^1\cap L^2(\R^3)$, the multiplication operator $W(x-y)$ on $L^2((\R^3)^2)$ satisfies
\begin{equation}  
0\le W(x-y)\le C \|W\|_{L^{3/2}(\R^3)} (-\Delta_x) \label{eq:W1} 
\end{equation}
and, for any $0 \leq \delta < 1/4$
\begin{equation}
0\le W(x-y)\le C_\delta \|W\|_{L^1(\R^3)} (1-\Delta_x)^{1-\delta}  (1-\Delta_y)^{1-\delta}. \label{eq:W2}  
\end{equation}
Moreover, for all $1>\eps>0$, $s>0$, $A \in L^3_{\rm loc}(\R^3,\R^3)$ and $0\le V\in L^1_{\rm loc}(\R^3)$, 
\begin{multline} 
\wh_x W(x-y) +W(x-y)\wh_x \\
\ge - C\left( \| W \|_{L^2} + (1+s^{2}) \|W\|_{L^{3/2}} \right) (1-\Delta_x) (1-\Delta_y).\label{eq:W3}
\end{multline}
\end{lemma}

\begin{proof}[Proof of Lemma \ref{lem:W}] We start with the \ 

\medskip
\noindent\textbf{Proof of~\eqref{eq:W1}.} From H\"older's and Sobolev's inequalities, we have 
\begin{align*}
 \langle f, W(x-y)f \rangle & =\iint W(x-y)  |f(x,y)|^2 dxdy \\
&\le \int \left( \int W(x-y)^{3/2}dx\right)^{2/3} \left( \int |f(x,y)|^{6}dx \right)^{1/3} dy \\
& \le C  \|W\|_{L^{3/2}(\R^3)} \int \left( \int |\nabla_x f(x,y)|^2 dx \right) dy 
\end{align*}
for every function $f\in H^{1}((\R^3)^2)$. Therefore, \eqref{eq:W1} follows immediately.

\medskip

\noindent\textbf{Proof of~\eqref{eq:W2}.} The estimate \eqref{eq:W2} with $\delta=0$ was first derived in \cite{ErdYau-01}. The following proof is adapted from the proof (again for $\delta=0$) in \cite{LieSei-06}. Note that for every operator $K$ we have $K^*K \le 1$ if and only if $KK^*\le 1$. Therefore, \eqref{eq:W2} is equivalent to 
\begin{equation} \label{eq:W2-equiv}
\sqrt{W(x-y)} (1-\Delta_x)^{\delta-1}(1-\Delta_y)^{\delta-1} \sqrt{W(x-y)}\le C_\delta \|W\|_{L^1}.
\end{equation}
Let $G$ be the Green function of $(1-\Delta)^{\delta-1}$, whose Fourier transform is given by 
$$ \widehat G (k) := \int_{\R^3} e^{-2\pi i x.k} G(x)dx =  \frac{1}{(1+4\pi^2|k|^2)^{1-\delta}}. $$
For every function $f\in L^2((\R^3)^2)$ one has
\begin{align*}
&\langle f, \sqrt{W(x-y)} (1-\Delta_x)^{\delta-1}  (1-\Delta_y)^{\delta-1} \sqrt{W(x-y)} f\rangle \\
 =& \int   \overline{f(x,y)}  \sqrt{W(x-y)} G(x-x') G(y-y') \sqrt{W(x'-y')} f(x',y') dxdydx'dy' \\
 &\le  \int  \frac{W(x-y) |G(x-x')|^2 |f(x',y')|^2 + W(x'-y') G(y-y')|^2 |f(x,y)|^2  }{2}  \\
& \quad =  C_\delta \|W\|_{L^1} \langle f, f \rangle
\end{align*}
where
$$
C_\delta := \int |G|^2 = \int |\hat G|^2 =  \int_{\R^3} \frac{dk}{(1+4\pi^2|k|^2)^{2(1-\delta)}} 
$$
which is finite for all $0 \leq \delta < 1/4$. Thus \eqref{eq:W2-equiv}, and hence \eqref{eq:W2}, holds true.

\medskip

\noindent\textbf{Simpler version of~\eqref{eq:W3}.} We are going to deduce~\eqref{eq:W3} from the inequality
\begin{multline} \label{eq:W3-a} 
 (-\Delta_x) W(x-y) + W(x-y) (-\Delta_{x}) \\
\ge -C \Big( \|W\|_{L^{3/2}} + \|W\|_{L^{2}} \Big) (1-\Delta_x) (1-\Delta_y) .
\end{multline} 
By an approximation argument, one can assume that $W$ is smooth. For every $f\in H^2(\R^3\times \R^3)$, a straightforward calculation using integration by part, and the identity $\nabla_x (W(x-y))=-\nabla_y (W(x-y))$ gives us
\begin{multline*}
 \left\langle f, \Big((-\Delta_x) W(x-y) + {W(x-y)} (-\Delta_x) \Big) f \right\rangle \\
= 2\Re \iint \nabla_x \overline{f(x,y)} \nabla_x (W(x-y)f(x,y)) dxdy \\
 =2\iint |\nabla_x {f(x,y)}|^2 W(x-y) + 2\Re \iint \nabla_x \overline{f(x,y)} \nabla_x (W(x-y)) f(x,y) dxdy \\
\ge   - 2\Re \iint \nabla_x \overline{f(x,y)} \nabla_y (W(x-y)) f(x,y) dxdy \\
= 2\Re \iint \nabla_y \Big( \big( \nabla_x \overline{f(x,y)}\big) f(x,y)\Big)  W(x-y)  dxdy \\
 = 2\Re \iint \left[ \nabla_x \overline{f(x,y)} \nabla_y f(x,y) + \nabla_y \nabla_x \overline{f(x,y)}\big) f(x,y) \right]W(x-y) dxdy.
\end{multline*}
Using Cauchy-Schwarz and Sobolev's inequality~\eqref{eq:W1}, we get
\begin{align*}
&\left| \iint \nabla_x \overline{f(x,y)} \nabla_y f(x,y)W(x-y) dxdy \right| \\
\le &  \iint \frac{|\nabla_x {f(x,y)}|^2+ |\nabla_y {f(x,y)}|^2}{2} |W(x-y)| dxdy \\
\le &  C\|W\|_{L^{3/2}}\langle f, (-\Delta_x)(-\Delta_y) f \rangle.  
\end{align*}
Moreover, by the Cauchy-Schwarz inequality again and \eqref{eq:W2} (with $\delta=0$ and $W$ replaced by $W^2$),
\begin{align*}
&\Big|  \iint \big( \nabla_y \nabla_x \overline{f(x,y)}\big) f(x,y)W(x-y)  dxdy \Big|\\
&\le \left(\iint |\nabla_y \nabla_x f(x,y)|^2 dxdy \right)^{1/2} \left(\iint |f(x,y)|^2 |W(x-y)|^2 dxdy \right)^{1/2} \\
& \le C\|W\|_{L^2} \langle f, (1-\Delta_x)(1-\Delta_y) f \rangle.  
\end{align*}
Thus we obtain 
\begin{align*}
& \left\langle f, \Big((-\Delta_x) W(x-y) + W(x-y) (-\Delta_x) \Big) f \right\rangle \nn\\
\ge & -C \Big( \|W\|_{L^{3/2}} + \|W\|_{L^{2}} \Big) \langle f, (1-\Delta_x) (1-\Delta_y) f \rangle
\end{align*}
for all $f\in H^2(\R^3\times \R^3)$. This proves~\eqref{eq:W3-a}. 

\medskip

\noindent\textbf{Proof of~\eqref{eq:W3}.} From the commutator relation
$$ p_x W(x-y)= W(x-y)p_x + (-i\nabla_x W)(x-y)$$
we find that
\begin{align*}
&(p_x A(x)+A(x)p_x + |A(x)|^2 )W(x-y)\nn\\
&\quad +W(x-y)(p_x A(x)+A(x)p_x+|A(x)|^2)\nn\\
&= 2\Big( p_x  W(x-y) A(x) + A(x) W(x-y)p_x) + |A(x)|^2 W(x-y) \Big) \nn \\
&= 2(p_x + A(x))W(x-y)(p_x+A(x)) - 2 p_x W(x-y)p_x.
\end{align*}
Using 
$$ (p_x + A(x))W(x-y)(p_x+A(x)) \ge 0$$
and estimating $p_x W(x-y)p_x$ by Sobolev's inequality~\eqref{eq:W1}, we get
\begin{align} \label{eq:W3-b}
&(p_x A(x)+A(x)p_x + |A(x)|^2 )W(x-y)\nn\\
&\quad +W(x-y)(p_x A(x)+A(x)p_x+|A(x)|^2)\nn\\
&\ge - C\|W\|_{L^{3/2}} (-\Delta_x)(-\Delta_y) .
\end{align}
Finally, by \eqref{eq:W1} again and the Cauchy-Schwarz inequality for operators
\begin{align} \label{eq:Cauchy-Schwarz} \pm ( XY+ Y^* X^* ) \le \delta XX^* + \delta^{-1} Y^*Y, \quad \forall \delta>0,
\end{align}
we obtain 
\begin{align*}
&  p_x^2(1-\theta_s( p_x)) W(x-y)+ W(x-y)p_x^2(1-\theta_s( p_x))  \\
&\ge - \delta p_x^2 (1-\theta_s( p_x)) W(x-y) p_x^2 (1-\theta_s( p_x)) + \delta^{-1} W(x-y) \\
&\ge - C \|W\|_{L^{3/2}} \Big( \delta p_x^4 (1-\theta_s( p_x))^2 + \delta^{-1}\Big)(-\Delta_x)
\end{align*}
for all $\delta>0$. Using $1-\theta_s( p) \le \1( |p|\le 2s)$ and choosing $\delta\sim s ^{-2}$ gives 
\begin{align}\label{eq:W3-c}
& p_x^2(1-\theta_s( p_x)) W(x-y)+ W(x-y)p_x^2(1-\theta_s( p_x)) \nn\\
&\ge  -Cs^{2}\|W\|_{L^{3/2}} (-\Delta_x).
\end{align}
From \eqref{eq:W3-a}, \eqref{eq:W3-b} and \eqref{eq:W3-c}, the bound \eqref{eq:W3} follows.
\end{proof}

\subsection{Proof of Lemma~\ref{lem:a-priori-estimate}}\label{sec:proof a priori}

Before completing the proof of Lemma~\ref{lem:a-priori-estimate} we make a remark on the simpler case with the Dyson potential $W_N$ replaced by a truly two-body interaction.

\begin{remark}[Second moment estimate with two-body interactions]\label{rem:simpler mom bound}\mbox{}\\ 
Consider the model case 
$$
K_N := \sum_{j=1}^N \wh_j + \frac{(1-\eps)^2}{N} \sum_{i\ne j}U_R(x_i-x_j) .
$$
By expanding $K_N^2$ and using the fact that $\wh_i\geq 0$ commutes with $U_R (x_j-x_k)\geq 0$ when $i\neq j$ and $i\neq k$, and then using \eqref{eq:W3} to estimate terms of the form 
$$\wh_i U_R(x_i-x_j)+U(x_i-x_j)\wh_i,$$
we obtain 
\begin{equation}\label{eq:simpler}
K_N ^2 \geq \frac{1}{3} \sum_{1 \leq i \neq j \leq N} \wh_i \wh_j 
\end{equation}
provided that $R=R(N)\gg N ^{-2/3}$. A similar estimate also holds when $\wh$ is replaced by the original kinetic operator $h$.  
\end{remark}

We stress once again that we do \emph{not} expect~\eqref{eq:simpler} to hold for our original Hamiltonian $H_N$, which is in the more singular regime $R\sim N ^{-1}$.  We thus need to work with the Dyson Hamiltonian, and its rather intricate nature makes the actual proof of Lemma~\ref{lem:a-priori-estimate} more difficult than the one we have sketched for \eqref{eq:simpler}. We now proceed with this proof.

\begin{proof}[Proof of Lemma \ref{lem:a-priori-estimate}] 
We have
\begin{align} \label{eq:h1h2-0}
&(\wH_N)^2  -  \Big(\sum_{j=1}^N \wh_j\Big)^2 = \frac{(1-\eps)^2}{N} \sum_{\ell=1}^N (\wh_\ell W_N + W_N \wh_\ell)  +\frac{(1-\eps)^4}{N^2} W_N^2.
\end{align}
Similarly as in Remark~\ref{rem:simpler mom bound}, the goal is to bound $\wh_1 W_N + W_N \wh_1$ from below. We first decompose the interaction operator as 
$$
W_N = W_a + W_b
$$
where
\begin{align*} 
W_a &= \sum_{1\in \{i,j\}} U_R(x_i-x_j) \prod_{k\ne i,j} \theta_{2R}(x_j-x_k),\\
W_b &=  \sum_{i,j \ge 2} U_R(x_i-x_j) \prod_{k\ne i,j} \theta_{2R} (x_j-x_k).
\end{align*}
{\bf Estimate of $W_a$.} By the Cauchy-Schwarz inequality \eqref{eq:Cauchy-Schwarz} we get
\begin{align} \label{eq:T1-1}
\pm(\wh_1 W_a+ W_a \wh_1) \le N^{-1} \wh_1 W_a \wh_1 + N W_a. 
\end{align}
Let us show that
\begin{align} \label{eq:T1-3} W_a \le \frac{C}{R^3}.
\end{align} 
Indeed, for every given $(x_1,x_2,...,x_N)\in (\R^3)^N$, the product 
$$ U_R(x_1-x_j) \prod_{k\ne 1,j} \theta_{2R}(x_j-x_k)$$
is bounded by $\|U_R\|_{L^\infty}\le CR^{-3}$ and it is zero except in the case 
$$ |x_1-x_j|< R < 2R < \min_{k\ne 1,j}|x_j-x_k|.$$
By the triangle inequality, the latter condition implies that 
\begin{align*}
|x_1-x_j|< R < \min_{k\ne 1,j}|x_1-x_k|
\end{align*}
and it is satisfied by at most one index $j \ne 1$. Therefore,
$$ \sum_{j\ge 2} U_R(x_1-x_j) \prod_{k\ne 1,j} \theta_{2R}(x_j-x_k) \le \frac{C}{R^3}.$$
Similarly, we have
$$ \sum_{i\ge 2} U_R(x_i-x_1) \prod_{k\ne 1,i} \theta_{2R}(x_1-x_k) \le \frac{C}{R^3}$$
and hence \eqref{eq:T1-3} holds true. From \eqref{eq:T1-1} and \eqref{eq:T1-3} we obtain
\begin{align} \label{eq:T1} \pm  \Big(\wh_1 W_a + W_a \wh_1\Big) \le \frac{C}{NR^3} (\wh_1)^2 + 2N \sum_{1\in \{i,j\}} U_R(x_i-x_j) \prod_{k\ne i,j} \theta_{2R}(x_j-x_k).
\end{align}
Here we do not need to estimate the second term on the right side of \eqref{eq:T1} because this term is part of $W_N$ which will be controlled by $W_N^2$ in $\wH_N^2$. 
\medskip
\text{}\\
{\bf Estimate of $W_b$.} We need a further decomposition
$$
W_b  = \sum_{i,j \ge 2} U_R(x_i-x_j) \prod_{k\ne i,j} \theta_{2R} (x_j-x_k) = W_c - W_d
$$
where
\begin{align*}
W_c &:=  \sum_{i,j \ge 2} U_R(x_i-x_j) \prod_{k\ne 1,i,j} \theta_{2R} (x_j-x_k) \\
W_d &:=  \sum_{i,j \ge 2} U_R(x_i-x_j) \Big(1-\theta_{2R}(x_j-x_1)\Big) \prod_{k\ne 1,i,j} \theta_{2R} (x_j-x_k) .
\end{align*}
Note that 
$$W_c\ge 0, \: W_d\ge 0 \mbox{ and } \wh_1W_c = W_c \wh_1 \ge 0.$$
On the other hand, by the Cauchy-Schwarz inequality \eqref{eq:Cauchy-Schwarz} again,
\begin{align} \label{eq:T2-1}
 \pm ( \wh_1 W_d + W_d \wh_1) \le \delta \wh_1 W_d \wh_1 + \delta^{-1} W_d.
\end{align}
We have two different ways to bound $W_d$. First, by \eqref{eq:W1} and \eqref{eq:wh-lower},
\begin{align*}  \Big(1-\theta_{2R}(x_j-x_1)\Big) \le C\| 1-\theta_{2R} \|_{L^{3/2}}(1-\Delta_1) \le C_{\eps} R^2 \wh_1.
\end{align*} 
Since here $i,j\geq 2$, both sides of the latter estimate commute with 
$$ U_R(x_i-x_j) \prod_{k\ne 1,i,j} \theta_{2R} (x_j-x_k),$$
and we deduce that  
\begin{align*}
&  \Big(1-\theta_{2R}(x_j-x_1)\Big) U_R(x_i-x_j) \prod_{k\ne 1,i,j} \theta_{2R} (x_j-x_k) \\
&\le C_{\eps} R^2 \wh_1 U_R(x_i-x_j) \prod_{k\ne 1,i,j} \theta_{2R} (x_j-x_k).
\end{align*}
Taking the sum over $i,j\ge 2$ we obtain  
\begin{align} \label{eq:T2-2}
W_d  \le C_{\eps}R^2 \wh_1 W_c.
\end{align}
Second, let us show that
\begin{align}\label{eq:T2-3} W_d \le \frac{C}{R^{3}}.
\end{align}
Indeed, for every given $(x_1,x_2,...,x_N)\in (\R^3)^N$, the product 
$$
U_R(x_i-x_j) \Big(1-\theta_{2R}(x_j-x_1)\Big) \prod_{k\ne 1,i,j} \theta_{2R} (x_j-x_k)
$$
is zero except in the case
\begin{align}\label{eq:packing} |x_i-x_j|< R, |x_j-x_1|< 4R, \min_{k\ne 1,i,j}|x_j-x_k|>2R.
\end{align}
By the triangle inequality, \eqref{eq:packing} implies that the ball $B(x_1,5R)$ contains $B(x_i,R/2)$, $B(x_j,R/2)$, and the balls $B(x_i,R/2)$, $B(x_j,R/2)$ do not intersect with $B(x_k,R/2)$ for all $k\ne 1,i,j$. Since $B(x_1,5R)$ can contain only a finite number of disjoint balls of radius $R/2$, we see that there are only a finite number of pairs $(i,j)$ satisfying \eqref{eq:packing}.  Thus we can conclude that
\begin{align*} W_d \le C \|U_R\|_{L^\infty} \le CR^{-3}.
\end{align*}
From \eqref{eq:T2-1}, \eqref{eq:T2-2} and \eqref{eq:T2-3}, we obtain
\begin{align*}
\wh_1 W_b+ W_b \wh_1 &= \wh_1 W_d+ W_d \wh_1 + 2 \wh_1 W_c  \\
&\ge - \frac{C\delta}{R^{3}} (\wh_1)^2 + \left(2-\frac{C_{\eps}R^2}{\delta} \right) \wh_1 W_c . 
\end{align*}
Choosing $\delta \sim R^2$ we get
\begin{align} \label{eq:T2}
\wh_1 W_b + W_b \wh_1 \ge - \frac{C_{\eps}}{R} (\wh_1)^2 .
\end{align}
{\bf Conclusion.} From \eqref{eq:T1} and \eqref{eq:T2}, we get
\begin{align*} 
\wh_1 W_N + W_N \wh_1 &\ge -\left( \frac{C}{NR^3} + \frac{C_{\eps}}{R} \right)  (\wh_1)^2 \\
&\quad - 2N \sum_{1\in \{i,j\}} U_R(x_i-x_j) \prod_{k\ne i,j} \theta_{2R}(x_j-x_k). 
\end{align*}
Summing the similar estimates with $1$ replaced by $\ell$ and using 
$$
\sum_{\ell=1}^N \sum_{\ell \in \{i,j\}} U_R(x_i-x_j) \prod_{k\ne i,j} \theta_{2R}(x_j-x_k) = 2W_N
$$
we find that
\begin{align*} 
\sum_{\ell=1}^N \Big( \wh_\ell W_N + W_N \wh_\ell \Big) &\ge -\left( \frac{C}{NR^3} + \frac{C_{\eps}}{ R} \right)  \sum_{\ell=1}^N (\wh_\ell)^2 - 2NW_N. 
\end{align*}
Therefore, coming back to \eqref{eq:h1h2-0} we conclude that (completing a square in the last inequality)
\begin{align*} 
&(\wH_N)^2  -  \Big(\sum_{j=1}^N \wh_j\Big)^2 = \frac{(1-\eps)^2}{N}\sum_{\ell=1}^N \Big( \wh_\ell W_N + W_N \wh_\ell \Big) + \frac{(1-\eps)^4}{N^2}W_N^2 \\
&\ge -\left( \frac{C}{N^2R^3} + \frac{C_{\eps}}{N R} \right)  \sum_{\ell=1}^N (\wh_\ell)^2 - 2(1-\eps)^2 W_N + \frac{(1-\eps)^4}{N^2}W_N^2 \\
&\ge -\left( \frac{C}{N^2R^3} + \frac{C_{\eps}}{N R} \right)  \sum_{\ell=1}^N (\wh_\ell)^2 - N^2.
\end{align*}
When $R\gg N^{-2/3}$ we have 
$$ \frac{C}{N^2R^3} + \frac{C_{\eps}}{N R} \ll 1$$
and hence
\begin{equation*} 
(\wH_N)^2  \geq 2 \sum_{1\leq i < j \leq N} \wh_i \wh_j + (1-o(1)) \sum_{\ell=1}^N (\wh_\ell)^2 - N^2   
\end{equation*}
which yields the result, recalling that in our convention $\wh \geq 1$.
\end{proof}

\subsection{Three-body estimate}\label{sec:3-body}

A first consequence of the second moment estimate in Lemma \ref{lem:a-priori-estimate} is that we can conveniently bound Dyson's Hamiltonian from below by a two-body Hamiltonian. This is done by first using a simple bound in terms of a three-body Hamiltonian, and then bounding the unwanted three-body part.

\begin{lemma}[\textbf{Three-body estimate}]\label{lem:3-body} \mbox{}\\
Assume the extra condition \eqref{eq:assumption-A-V} holds. For every $1>\eps>0$ and $s>0$, if  $R=R(N) \gg N^{-2/3}$, then
\begin{align}\label{eq:3-body} \sum_{i\ne j} U_R(x_i-x_j)\sum_{k\ne i,j} (1-\theta_{2R}(x_j-x_k)) \le C_{\eps,s}\frac{R^2}{N} (\wH_N)^4. 
\end{align}
Consequently, 
\begin{align}\label{eq:cut-off-remove}
\wH_N \ge \sum_{j=1}^N \wh_j + \frac{(1-\eps)^2}{N} \sum_{i\ne j}U_R(x_i-x_j) - C_{\eps,s}\frac{R^2}{N^{2}} (\wH_N)^4.
\end{align}
\end{lemma}
Note the error term involving $(\wH_N)^4$, which is well under control since we are interested in its expectation value in a ground state.
\begin{proof} By \eqref{eq:W1} and \eqref{eq:wh-lower} we have 
$$(1-\theta_{2R}(x_2-x_k))\le C_{\eps,s}R^2 \wh_k\mbox{ for } k\ge 3.$$
Since $U_R(x_1-x_2)$ commutes with both sides, we get
\begin{align} \label{eq:3-body-a}
& U_R(x_1-x_2)\sum_{k\ge 3} (1-\theta_{2R}(x_2-x_k)) \le C_{\eps,s}R^2 U_R(x_1-x_2)\sum_{k\ge 3} \wh_3 \nn\\
&= \frac{1}{2}C_{\eps,s} R^2  \Big( \wH_N - \wh_1 -\wh_2 - (1-\eps)^2 N^{-1}W_N \Big)U_R(x_1-x_2) \nn\\
& \quad +\frac{1}{2}C_{\eps,s} R^2 U_R(x_1-x_2)\Big( \wH_N - \wh_1 -\wh_2 - (1-\eps)^2 N^{-1}W_N \Big)  \nn\\
&\le \frac{1}{2}C_{\eps,s} R^2 \left( \wH_N U_R(x_1-x_2)  +  U_R(x_1-x_2)\wH_N \right) \nn\\
&\quad +  \frac{1}{2} C_{\eps,s} R^2 \sum_{j=1}^2 \left( \wh_j U_R(x_1-x_2)  + U_R(x_1-x_2) \wh_j \right)\,.
\end{align}
In the last estimate we have used $W_N \ge 0$. Thanks to \eqref{eq:W3} and \eqref{eq:wh-lower}, we get for all $j=1,2,$
\begin{align} \label{eq:3-body-c}
&\wh_j U_R(x_1-x_2)+U_R(x_1-x_2)h_j \nn\\
&\ge -C_{\eps,s} R^{-3/2} (1-\Delta_1)(1-\Delta_2) \ge - C_{\eps,s} R^{-3/2}\wh_1 \wh_2.
\end{align}
On the other hand, by  the Cauchy-Schwarz inequality \eqref{eq:Cauchy-Schwarz} and \eqref{eq:W2} (with $\delta=0$ and $W=U_R$) 
\begin{align}\label{eq:3-body-d} &\wH_N  U_R(x_1-x_2) + U_R(x_1-x_2) \wH_N \nn \\
&\le \delta \wH_N  U_R(x_1-x_2) \wH_N + \delta^{-1} U_R(x_1-x_2) \nn\\
&\le C_{\eps,s} \delta \wH_N \wh_1 \wh_2 \wH_N + C_{\eps,s} \delta^{-1} \wh_1 \wh_2
\end{align}
for all $\delta>0$. Choosing $\delta=N^{-1}$ and using $R^{-3/2}\le N$, we deduce from \eqref{eq:3-body-a}, \eqref{eq:3-body-c} and \eqref{eq:3-body-d} that
\begin{align*}
U_R(x_1-x_2)\sum_{k\ge 3} (1-\theta_{2R}(x_2-x_k)) 
\le C_{\eps,s} R^2 \Big( N^{-1}\wH_N \wh_1 \wh_2 \wH_N +  N \wh_1 \wh_2 \Big). 
\end{align*}
By symmetrization with respect to the indices, we find that
\begin{align*}
&\sum_{i\ne j}U_R(x_1-x_2)\sum_{k\ne i,j} (1-\theta_{2R}(x_j-x_k)) \\
&\le C_{\eps,s} R^2 \Big( N^{-1}\wH_N \sum_{i\ne j}\wh_i \wh_j \wH_N +  N \sum_{i\ne j} \wh_i \wh_j \Big) . 
\end{align*}
Combining with the second moment estimate \eqref{eq:h1h2} we obtain \eqref{eq:3-body}. From the three-body estimate \eqref{eq:3-body} and the elementary inequality \eqref{eq:WN-UR}, the operator bound \eqref{eq:cut-off-remove} follows. 
\end{proof}

\section{Energy lower bound and convergence of states}\label{sec:low bound}

\subsection{Mean-field approximation} \label{sec:MF}
We are now reduced to justifying the mean-field approximation for a new Hamiltonian with the two-body interaction $U_R(x-y)$ which converges to a Dirac delta much slower than the original one. 
The analysis in this section provides an alternative to the coherent states method of~\cite{LieSei-06}. 

\begin{proposition}[\textbf{Mean-field approximation}]\label{pro:MF}\mbox{}\\
Assume that~\eqref{eq:assumption-A-V} holds. For every $1>\eps>0$ and $s>0$, if 
$$N^{-1/2} \gg R=R(N) \gg N^{-2/3}$$
then
\begin{align} \label{eq:MF} \lim_{N\to \infty} \frac{\inf\sigma(\wH_N)}{N} =  \inf_{\|u\|_{L^2}=1} \left( \langle u, \wh u \rangle +  (1-\eps)^2 4\pi a \int |u|^4 \right)=:e_{{\rm NL}}(\eps,s). 
\end{align}
\end{proposition}

\begin{proof} The upper bound in \eqref{eq:MF} can be obtained easily using trial states of the form $u^{\otimes N}$. For the lower bound, let us consider a ground state $\wPsi_N$ of $\wH_N$ (which exists because $\wh$ has compact resolvent). Using the ground state equation, we find that 
\begin{align} \label{eq:wPsi-Hk}\langle \wPsi_N, (\wH_N)^k \wPsi_N\rangle=(\inf \sigma(\wH_N))^k\le (C_{\eps,s}N)^k
\end{align}
for all $k\in \mathbb{N}$. In particular, the second moment estimate \eqref{eq:h1h2} implies that
\begin{align}\label{eq:wPsi-h1h2} \langle \wPsi_N, \wh_1 \wh_2 \wPsi_N \rangle \le C_{\eps,s}
\end{align}
and the operator estimate \eqref{eq:cut-off-remove} implies that
\begin{align}\label{eq:cut-off-remove-b}
\liminf_{N\to \infty}\frac{\langle \wPsi_N,\wH_N \wPsi_N \rangle}{N} \ge \liminf_{N\to \infty} \left( \Tr \Big(\wh \gamma_{\wPsi_N}^{(1)}\Big)+ (1-\eps)^2\Tr \Big( U_R \gamma_{\wPsi_N}^{(2)}\Big) \right). 
\end{align}
Here $\gamma_{\wPsi_N}^{(k)}$ is the $k$-particle density matrices of $\wPsi_N$ and $U_R$ is understood as the multiplication operator $U_R(x-y)$ on $\gH^2$. Since $\Tr \Big( \wh \gamma_{\wPsi_N}^{(1)}\Big)$ is bounded uniformly in $N$ and $\wh$ has compact resolvent, up to a subsequence we can assume that $\gamma_{\wPsi_N}^{(1)}$ converges strongly in trace class. By the quantum de Finetti Theorem \ref{thm:DeFinetti}, up to a subsequence we can find a Borel probability measure $\wmu$ on the unit sphere $S\gH$ such that
\begin{align} \label{eq:deF-wPsi}
\lim_{N\to \infty} \Tr \left| \gamma_{\wPsi_N}^{(k)} - \int |u^{\otimes k}\rangle \langle u^{\otimes k}| d\wmu(u)\right| =0,\quad \forall k\in \mathbb{N}.
\end{align}
We will show that
\begin{align}\label{eq:lower-deF}
&\liminf_{N\to \infty} \left( \Tr \Big(\wh \gamma_{\wPsi_N}^{(1)}\Big)+ (1-\eps)^2\Tr \Big( U_R \gamma_{\wPsi_N}^{(2)}\Big) \right) \nn\\
&\ge \int \Big( \langle u, \wh u \rangle +  (1-\eps)^2 4\pi a \int |u|^4 \Big) d\wmu(u)
\end{align}
and then the lower bound in \eqref{eq:MF} follows immediately. Since $\wh$ is positive and independent of $N$, \eqref{eq:deF-wPsi}  and Fatou's lemma imply
\begin{align} \label{eq:deF-wPsi-h}
\liminf_{N\to \infty} \Tr \Big( \wh\gamma_{\wPsi_N}^{(1)} \Big) \ge \int \langle u, \wh u \rangle d\wmu(u).
\end{align}
It remains to prove
\begin{align}\label{eq:deF-wPsi-U}
\liminf_{N\to \infty} \Tr \Big( U_R \gamma_{\wPsi_N}^{(2)} \Big) \ge 4\pi a \int \|u\|^4_{L^4} d\wmu(u).
\end{align}
Note that \eqref{eq:deF-wPsi-U} does not follow from \eqref{eq:deF-wPsi} and Fatou's lemma easily because $U_R$ depends on $R=R(N)$. We need to replace $U_R$ by an  operator bounded independently of $N$. Since $\wh$ has compact resolvent, for every $\Lambda \ge 1$ the projection 
$$P_\Lambda:= \1(\wh \le \Lambda)$$
has finite rank. Let us denote 
$$ \Pi := \1_{\gH ^2} - P_\Lambda ^{\otimes 2}.$$ 
Since $U_R \geq 0$, we can apply the Cauchy-Schwarz inequality~\eqref{eq:Cauchy-Schwarz} with 
$$X= \PL^{\otimes 2} U_R ^{1/2} \mbox{ and } Y= U_R ^{1/2} \Pi$$
to obtain
\begin{align*}
U_R &= (\PL ^{\otimes 2} + \Pi) U_R (\PL ^{\otimes 2} + \Pi) \\
&= \PL ^{\otimes 2} U_R \PL ^{\otimes 2} + \Pi U_R \Pi + \PL ^{\otimes 2} U_R \Pi + \Pi U_R \PL ^{\otimes 2} \\
&\geq \PL ^{\otimes 2} U_R \PL ^{\otimes 2}  - \delta ^{-1} \Pi U_R \Pi - \delta \PL ^{\otimes 2} U_R \PL ^{\otimes 2}
\end{align*}
for all $\delta>0$. Using the operator bound \eqref{eq:W2} and the fact that the $4/5$-th power is operator monotone~\cite{Bhatia} we have 
\begin{align} \label{eq:UR-h1h2} 
U_R(x_1-x_2) \le C\|U_R\|_{L^1} (1-\Delta_1)^{4/5}(1-\Delta_2)^{4/5} \le C_{\eps,s} (\wh_1)^{4/5}(\wh_2)^{4/5}.
\end{align}
Therefore,
$$ \PL ^{\otimes 2} U_R \PL ^{\otimes 2} \le C_{\eps,s} \wh_1 \wh_2 \quad \text{and}\quad \Pi U_R \Pi \leq C_{\eps,s} \Lambda^{-1/5} \wh_1 \wh_2.$$
Here in the second estimate we have used $\1_{\gH } - \PL  \le \Lambda^{-1/5} (\wh)^{1/5}$, which is a consequence of the definition of $\PL$. Thus
$$
U_R - \PL^{\otimes 2}U_R \PL^{\otimes 2} \ge -C_{\eps,s} (\delta^{-1} + \delta \Lambda ^{-1/5}) \wh_1 \wh_2.
$$
If we choose $\delta = \Lambda ^{-1/10}$ and take the trace against $\gamma_{\wPsi_N}^{(2)}$, then by the a-priori estimate \eqref{eq:wPsi-h1h2} we find
\begin{align} \label{eq:UR-a}
\Tr \Big( U_R \gamma_{\wPsi_N}^{(2)} \Big) -  \Tr \Big( \PL^{\otimes 2}U_R \PL^{\otimes 2} \gamma_{\wPsi_N}^{(2)} \Big) \ge - C_{\eps,s} \Lambda^{-1/10}.
\end{align}
On the other hand, from \eqref{eq:UR-h1h2} and the definition of $\PL$, it follows that the operator norm $\| \PL^{\otimes 2}U_R \PL^{\otimes 2} \|$ is bounded uniformly in $N$ for fixed $\Lambda$. Therefore, the strong convergence \eqref{eq:deF-wPsi} implies that
\begin{align} \label{eq:UR-b}
\lim_{N\to \infty}\left( \Tr \Big( \PL^{\otimes 2}U_R \PL^{\otimes 2} \gamma_{\wPsi_N}^{(2)} \Big) - \int \left\langle (\PL u)^{\otimes 2}, U_R (\PL u)^{\otimes 2}\right \rangle d\wmu(u) \right) =0.
\end{align} 
Since the left side of \eqref{eq:deF-wPsi-h} is finite, every function $u$ in the support of $d\wmu$ belongs to the quadratic form domain $Q(\wh)$ of $\wh$ and hence $\PL u\to u$ strongly in $Q(\wh)$. Using the continuous embeddings $Q(\wh)\subset H^1 \subset L^4$, we get
$$
\lim_{\Lambda\to \infty} \lim_{R\to 0} \langle (\PL u)^{\otimes 2}, U_R (\PL u)^{\otimes 2}\rangle = \lim_{\Lambda\to \infty} \|\PL u\|_{L^4}^4 = \|u\|_{L^4}^4.
$$
By Fatou's lemma,
\begin{align}\label{eq:UR-c}\liminf_{\Lambda\to \infty}\liminf_{N\to \infty} \int \langle (\PL u)^{\otimes 2}, U_R (\PL u)^{\otimes 2}\rangle d\wmu(u) \ge 4\pi a \int \|u\|_{L^4}^4 d\wmu(u).
\end{align}
The desired convergence \eqref{eq:deF-wPsi-U} follows from \eqref{eq:UR-a}, \eqref{eq:UR-b} and \eqref{eq:UR-c}.  
\end{proof}

\begin{remark}[Mean-field approximation with two-body interactions]\label{rem:simpler MF}\mbox{}\\ 
From the preceding proposition we obtain easily the convergence \eqref{eq:Hartree-cv} mentioned in Section \ref{sec:overall} because $\wH_N\le K_N$. In fact, $K_N$ satisfies the second moment estimate (\ref{eq:simpler}) (cf. Remark~\ref{rem:simpler mom bound}), and hence \eqref{eq:Hartree-cv} can be proved directly. In particular, the method can be used to derive the energy asymptotics when the interaction potential is given by (\ref{eq:wbeta}); for $\beta<2/3$, Step 1 (and thus also Step 3) in the proof are not needed. One can also obtain some explicit error estimate in Proposition~\ref{pro:MF} and \eqref{eq:Hartree-cv} by using a quantitative version of the quantum de Finetti theorem as in \cite[Lemma~3.4]{LewNamRou-14c}.      
\end{remark}

\subsection{Convergence of ground state energy} \label{sec:GP-energy}

We now conclude the proof of the convergence of the ground state energy. There are two things left to do: remove the high momentum cut-off in the final effective functional, and relax the additional assumption~\eqref{eq:assumption-A-V}.

\begin{proof}[Proof of energy convergence \eqref{eq:cv-energy}] The upper bound in \eqref{eq:cv-energy} was proved in~\cite{Sei-03}. The proof of the lower bound  is divided into three steps.  

\medskip

\noindent{\bf Step 1.} We start with the simple case when the extra condition \eqref{eq:assumption-A-V} holds true. Recall that we are choosing 
$$N^{-1/2}\gg R=R(N) \gg N^{-2/3}.$$
From Lemma \ref{lem:Dyson} and Proposition \ref{pro:MF} it follows that for every $1>\eps>0$ and~$s>0$,
$$
\liminf_{N\to \infty} \frac{\inf \sigma(H_N)}{N} \ge \liminf_{N\to \infty} \left( \frac{\inf\sigma(\wH_N)}{N} + \kappa_{\eps,s} \right) = e_{{\rm NL}}(\eps,s)+\kappa_{\eps,s}. 
$$
Thus to obtain the lower bound in \eqref{eq:cv-energy}, it remains to show that 
\begin{align}\label{eq:Enls-ke} \lim_{\eps\to 0}\lim_{s\to \infty} (e_{\rm NL}(\eps,s) + \kappa_{\eps,s})=e_{\rm GP}. 
\end{align}
The upper bound in \eqref{eq:Enls-ke} is trivial as $\cE_{\rm NL}(u)+\kappa_{\eps,s} \le \cE_{\rm GP}(u)$. The lower bound in \eqref{eq:Enls-ke} can be done by a standard compactness argument provided in \cite{LieSei-06}. We recall this here for the reader's convenience. Let $u_{\eps,s}$ be a ground state for $e_{\rm NL}(\eps,s)$, namely
$$
e_{\rm NL}(\eps,s)= \langle u_{\eps,s}, \wh  u_{\eps,s}   \rangle  + (1-\eps)^2 4\pi a \int |u_{\eps,s}|^4.
$$
From \eqref{eq:wh-lower} it follows that $\langle u_{\eps,s}, (-\Delta+V) u_{\eps,s}\rangle$ is bounded uniformly in $s$. Since $-\Delta+V$ has compact resolvent, for every given $\eps>0$ there exists a subsequence $s_j \to \infty$ such that $u_{\eps,s_j}$ converges strongly in $L^2$ and pointwise (in both $p$-space and $x$-space) to a function $u_\eps$. By Fatou's lemma we have
\begin{align*}
&\liminf_{j\to \infty} \int |u_{\eps,s_j}(x)|^4 dx \ge \int |u_{\eps}(x)|^4 dx,\\
&\liminf_{j\to \infty} \int p^2(1-\theta_{s_j}(p))|\hat u_{\eps,s_j}(p)|^2 dp \ge \int p^2|\hat u_\eps(p)|^2 dp. 
\end{align*}
Next, using~\eqref{eq:assumption-A-V} as before we have 
$$\eps p^2+ pA+Ap+|A|^2 + V + C_\eps \ge 0$$
for some $C_\eps \geq 0$. Using Fatou's lemma again and the strong convergence in $L^2$ we deduce
\begin{multline*}
\liminf_{j\to \infty} \left\langle u_{\eps,s_j}, \Big(\eps p^2+ pA+Ap+|A|^2 + V + \kappa_{\eps,s} \Big) u_{\eps,s_j} \right\rangle \\
\ge \left\langle u_{\eps}, \Big(\eps p^2+ pA+Ap+|A|^2 + V \Big) u_{\eps} \right\rangle.
\end{multline*}
Combining these estimates, we get
\begin{align*} \liminf_{j\to \infty} \Big(e_{\rm NL}(\eps,s_j)+\kappa_{\eps,s_j}\Big) \ge \langle u_{\eps}, h u_{\eps}\rangle + (1-\eps)^2 4\pi a \int |u_{\eps}|^4 \ge (1-\eps)^2 e_{\rm GP}.
\end{align*}
Taking $\eps\to 0$ we obtain the lower bound in \eqref{eq:Enls-ke}.

\medskip

\noindent{\bf Step 2.} From now on we do not assume \eqref{eq:assumption-A-V}. Let us introduce the Hamiltonian
$$
H_{M,N}:=\sum_{j=1}^M h_j + \sum_{1\leq i<j \leq M} w_N(x_i-x_j)
$$
and denote $E(M,N)$ its (bosonic) ground state energy. In this step we will prove the lower bound in \eqref{eq:cv-energy} using the additional assumption 
\begin{align} \label{eq:ass-binding}
E(N,N)-E(N-1,N) \le C.
\end{align}
We will find a function $f:\R^3\to \R_+$ growing faster than $|A|$, namely 
\begin{equation}\label{eq:growth f}
\lim_{|x|\to \infty} \frac{|A(x)|}{f(x)}=0 
\end{equation}
such that for a ground state $\Psi_N$ for $H_N$ we have 
\begin{equation}\label{eq:bounded f}
\langle \Psi_N, f^2(x_1) \Psi_N \rangle \le C.
\end{equation}
Once this is achieved we get
$$
\frac{\inf\sigma(H_N)}{N} \ge \frac{\inf\sigma\Big(H_N + \eta \sum_{j=1}^N f^2(x_j)\Big)}{N} - C\eta
$$
for every $\eta>0$. Since the growth condition \eqref{eq:assumption-A-V} holds true with $V$ replaced by $V+\eta f^2$, we can apply the result in Step 1 to the Hamiltonian 
$$H_N+\eta \sum_{j=1}^N f^2(x_j)$$
for every given $\eta>0$. Then the lower bound in \eqref{eq:cv-energy} follows by taking $\eta\to 0$. 

Now we find such a function $f$. We will establish a simple binding inequality using an idea in \cite{Lieb-84}. From the ground state equation $H_{N,N} \Psi_N =E(N,N)\Psi_N$, it follows that  
\begin{align} \label{eq:GSE-binding}
E(N,N) \langle \Psi_N, f^2(x_N) \Psi_N\rangle = \Re \langle \Psi_N, f^2(x_N) H_{N,N} \Psi_N \rangle.
\end{align}
By the variational principle and \eqref{eq:ass-binding} we have 
$$H_{N,N}-h_N \ge H_{N-1,N}\ge E(N-1,N) \ge E(N,N)-C.$$
Note that $f^2(x_N)$ commutes with all terms in the latter inequality. If $f$ is bounded and sufficiently regular, we have the IMS-type formula
\begin{equation}\label{eq:IMS}
\frac{1}{2}(f^2 h + hf^2)=fhf-|\nabla f|^2 \ge Vf^2-|\nabla f|^2, 
\end{equation}
and we deduce from \eqref{eq:GSE-binding} that
\begin{align} \label{eq:binding-EN-EN-1}
\left\langle \Psi_N,\Big( V(x_N)f^2(x_N)-|\nabla f(x_N)|^2 - C f^2(x_N) \Big) \Psi_N \right\rangle \le 0.
\end{align}
Note that if we choose $f(x)=e^{b|x|}$ for some constant $b>0$, then \eqref{eq:growth f} follows from the assumption \eqref{eq:assumption-A}. Moreover, heuristically \eqref{eq:bounded f} follows  from \eqref{eq:binding-EN-EN-1} as $Vf^2$ grows faster than $|\nabla f|^2+Cf^2$. To make this idea rigorous, let us apply \eqref{eq:binding-EN-EN-1} with $f(x)$ replaced by
$$ g_r(x)=\exp\left[b\left[r-\left||x|-r \right| \right]_+\right].$$
Note that $g_r(x)=e^{b|x|}$ when $|x|\le r$ and $g_r(x)=1$ when $|x|\ge 2r$. We can thus apply~\eqref{eq:IMS} to $g_r$.

Moreover,  
\begin{align*}
Vg_r^2 - |\nabla g_r|^2 - C g_r^2 &\ge (V-b^2 -C) g_r^2 \\
&\ge g_r^2 -  (b^2+C+1)g_r^2 \1 \left(V\le b^2+C+1\right) \\
&\ge g_r^2 - C_0 
\end{align*}
for some constant $C_0$ independent of $r>0$. Here we have used the fact that $g_r^2 \1(V\le b^2+C+1)$ is bounded independently of $r>0$, which follows from the assumption $\lim_{|x|\to \infty}V(x)=+\infty$. Thus  \eqref{eq:binding-EN-EN-1} gives us
\begin{align*} 
\left\langle \Psi_N, g_r (x_N) \Psi_N \right\rangle \le C_0
\end{align*}
for all $r>0$. Taking $r\to \infty$ we obtain \eqref{eq:bounded f} with $f(x)=e^{b|x|}$.  
\medskip
\text{}\\
{\bf Step 3.} Now we explain how to remove the additional assumption \eqref{eq:ass-binding}. This can be done by following the strategy in \cite{LieSei-06}, which we recall quickly below for the reader's convenience. 

By choosing trial states $u^{\otimes N}$, we get the upper bound 
$$E(N,N)\le C_0 N$$
for some constant $C_0 > 2e_{\rm GP}$. For every $N\in \mathbb{N}$, we denote by $M=M(N)$ the largest integer $\le N$ such that 
\begin{equation}\label{eq:binding M}
E(M(N),N)-E(M(N)-1,N)\le C_0. 
\end{equation}
Then by the choice of $M(N)$ we obtain 
\begin{align} \label{eq:M-N-1} E(N,N)-E(M(N),N)\ge (N-M(N))C_0.
\end{align}
We can find a subsequence $N_j\to \infty$ such that $M(N_j)/N_j\to \lambda\in [0,1]$. Since~\eqref{eq:binding M} holds with $M = M(N_j)$ we can apply the result in Step 2 with $w$ replaced by $\lambda w$ and find that
\begin{align} \label{eq:M-N-2}
\liminf_{j\to \infty} \frac{E(M(N_j),N_j)}{M_j} \ge e_{\rm GP}(\lambda a) \ge \lambda e_{\rm GP}(a).
\end{align}
Here $e_{\rm GP}(\lambda a)$ is the Gross-Pitaevskii energy with $a$ replaced by $\lambda a$ and the last inequality in \eqref{eq:M-N-2} is obtained by simply ignoring part of the one-body energy in the corresponding Gross-Pitaevskii functional. From \eqref{eq:M-N-1} and \eqref{eq:M-N-2}, it follows that
\begin{align*}
e_{\rm GP}(a) \ge \liminf_{j\to \infty}\frac{E(N_j,N_j)}{N_j} &\ge \liminf_{j\to \infty} \left( \frac{E(M(N_j),N_j)}{N_j}+ C_0\frac{N_j - M(N_j)}{N_j} \right) \\
& \ge \lambda^2 e_{\rm GP}(a)+ C_0 (1-\lambda)\,.
\end{align*}
Since 
$$e_{\rm GP}(a)\le \lambda^2 e_{\rm GP}(a)+ 2(1-\lambda)e_{\rm GP}(a)$$
and $C_0> 2e_{\rm GP}(a)$, we must have $\lambda=1$. Thus $M(N)/N\to 1$ for the whole sequence and  
$$ \liminf_{N\to \infty} \frac{E(N,N)}{N} =\liminf_{j\to \infty}\frac{E(N_j,N_j)}{N_j} \ge e_{\rm GP}(a).$$
This completes the proof of the energy convergence \eqref{eq:cv-energy}.
\end{proof}

\subsection{Convergence of density matrices} \label{sec:GS}

Now we prove the convergence of ground states in \eqref{eq:cv-state} by means of the Hellmann-Feynman principle. For  $v\in L^2(\R^3)$ and $\ell\in \mathbb{N}$ we will perturb $H_N$ by 
$$
S_{v,\ell}:=\frac{\ell!}{N^{\ell-1}}\sum_{1\le i_1<...<i_\ell\le N} |v^{\otimes \ell}\rangle \langle v^{\otimes \ell}|_{i_1,...,i_\ell}.
$$
Here $|v^{\otimes \ell}\rangle \langle v^{\otimes \ell}|_{i_1,...,i_\ell}$ denotes the operator $|v^{\otimes \ell}\rangle \langle v^{\otimes \ell}|$ acting on the $\ell$-body Hilbert space of the $i_1$-th,..., $i_\ell$-th variables. We have the following extension of  \eqref{eq:cv-energy}.

\begin{lemma}[\textbf{Energy lower bound for perturbed Hamiltonians}] \label{lem:lower-Hvk}\mbox{}\\
We assume \eqref{eq:assumption-A}, \eqref{eq:assumption-V} and \eqref{eq:assumption-w}. For every $v\in L^2(\R^3)$ and $\ell\in \mathbb{N}$, we have
\begin{align} \label{eq:lower-Hvk}
\liminf_{N\to \infty}\frac{\inf \sigma(H_N- S_{v,\ell})}{N}   \ge \inf_{\|u\|_{L^2}=1} \left( \cE_{\rm GP}(u)- |\langle v, u\rangle|^{2\ell}  \right).
\end{align}  

\end{lemma}

\begin{proof} 
We first work under the extra condition \eqref{eq:assumption-A-V}, and then explain how to remove it at the end. Let $1>\eps>0$ and $s>0$ and 
$$N^{-1/2}\gg R=R(N) \gg N^{-2/3}.$$
Recall that from \eqref{eq:Dyson-lemma} we have
\begin{align}\label{eq:Dyson-Phi}
H_N - S_{v,\ell} \ge \wH_N - S_{v,\ell} + N\kappa_{\eps,s} -  C_{\eps,s} N R^2.
\end{align}
Let $\Phi_N$ be a ground state for $\wH_N - S_{v,\ell}$. Since $\|S_{v,\ell}\|/N$ is bounded uniformly in $N$, \eqref{eq:wPsi-Hk} still holds true with $\wPsi_N$ replaced by $\Phi_N$, namely 
\begin{align} \label{eq:higher-moment-PhiN} \langle \Phi_N, (\wH_N)^k \Phi_N \rangle \le (C_{\eps,s} N)^k
\end{align}
for all $k\in \mathbb{N}$. Combining \eqref{eq:higher-moment-PhiN} with the three-body estimate in Lemma \ref{lem:3-body} we get the following analogue of \eqref{eq:cut-off-remove-b} 
\begin{align}\label{eq:cut-off-remove-b-Phi}
&\liminf_{N\to \infty}\frac{\inf \sigma(\wH_N-S_{v,\ell})}{N}  = \liminf_{N\to \infty}\frac{\langle \Phi_N,(\wH_N-S_{v,\ell}) \Phi_N \rangle}{N} \\
& \ge \liminf_{N\to \infty} \left( \Tr \Big(\wh \gamma_{\wPsi_N}^{(1)}\Big)+ (1-\eps)^2\Tr \Big( U_R \gamma_{\wPsi_N}^{(2)}\Big) - \Tr \Big( |v^{\otimes \ell}\rangle \langle v^{\otimes \ell}| \gamma_{\Phi_N}^{(\ell)}\Big) \right). \nn
\end{align}
Moreover, \eqref{eq:higher-moment-PhiN} and the second moment estimate \eqref{eq:h1h2} imply the a-priori estimate $ \langle \Phi_N, \wh_1 \wh_2 \Psi_N \rangle \le C_{\eps,s}$. Therefore, we can estimate the right side of \eqref{eq:cut-off-remove-b-Phi} by proceeding exactly as in the proof of Proposition~\ref{pro:MF}. More precisely, by the quantum de Finetti Theorem \ref{thm:DeFinetti}, we can find a Borel probability measure $\mu_\Phi$ on the unit sphere $S\gH$ such that, up a subsequence,
\begin{align*} 
\lim_{N\to \infty} \Tr \left| \gamma_{\Phi_N}^{(k)} - \int |u^{\otimes k}\rangle \langle u^{\otimes k}| d\mu_\Phi(u)\right| =0,\quad \forall k\in \mathbb{N}.
\end{align*}
Using \eqref{eq:lower-deF} with $\wPsi_N$ replaced by $\Phi_N$ and employing the fact that $|v^{\otimes \ell}\rangle \langle v^{\otimes \ell}|$  is bounded, we obtain 
\begin{align} \label{eq:lower-deF-Phi}
&\liminf_{N\to \infty} \left( \Tr \Big(\wh \gamma_{\wPsi_N}^{(1)}\Big)+ (1-\eps)^2\Tr \Big( U_R \gamma_{\wPsi_N}^{(2)}\Big) - \Tr \Big( |v^{\otimes \ell}\rangle \langle v^{\otimes \ell}| \gamma_{\Phi_N}^{(\ell)}\Big) \right) \nn\\
&\ge \int \Big( \langle u, \wh u \rangle +  (1-\eps)^2 4\pi a \int |u|^4 - |\langle v, u \rangle|^{2\ell}\Big) d\mu_\Phi(u).
\end{align}
From \eqref{eq:Dyson-Phi}, \eqref{eq:cut-off-remove-b-Phi} and \eqref{eq:lower-deF-Phi}, it follows that
\begin{align*}
&\liminf_{N\to \infty}\frac{\inf\sigma(H_{N}-S_{v,\ell})}{N} \\
&\ge \inf_{\|u\|_{L^2}=1} \left( \langle u, \wh u \rangle +  (1-\eps)^2 4\pi a \int |u|^4 - |\langle v, u \rangle|^{2\ell} \right) + \kappa_{\eps,s}.
\end{align*}
The lower bound \eqref{eq:lower-Hvk} follows by passing to the limits $s\to 0$ and then $\eps\to 0$ as in the proof of \eqref{eq:Enls-ke}.

To remove the assumption \eqref{eq:assumption-A-V} we may use the argument in Subsection \ref{sec:GP-energy}. The only extra difficulty is that when dealing with the analogue of \eqref{eq:GSE-binding} with $H_{N,N}$ replaced by $H_{N,N}-S_{v,\ell}$, we have to take care of the operator $f^2 |v \rangle \langle v|= |f^2 v \rangle \langle v|$ which may be unbounded as $f(x)=e^{b|x|}$ with $b>0$ and $v$ is merely in $L^2(\R^3)$. However, we can still proceed with all functions $v$ in $L^2(\R^3)$ which have compact support. Then after obtaining the lower bound \eqref{eq:lower-Hvk} with those nice functions $v$, we can extend the lower bound to all functions $v$ in $L^2(\R^3)$ by a standard density argument. 
\end{proof}

Now we are able to prove the convergence of density matrices.

\begin{proof}[Proof of state convergence \eqref{eq:cv-state}] Let $\Psi_N$ be an approximate ground state for $H_N$ as in Theorem \ref{thm:cv-GP}. For every $v\in L^2(\R^3)$ and $\ell\in \mathbb{N}$, from the upper bound in \eqref{eq:cv-energy} and the lower bound in Lemma \ref{lem:lower-Hvk} we have
\begin{align*}
&\limsup_{N\to \infty}\Tr \Big(|v^{\otimes \ell}\rangle \langle v^{\otimes \ell}| \gamma_{\Psi_N}^{(\ell)} \Big) \\ &= \limsup_{N\to \infty} \left(  \frac{\langle \Psi_N, H_N \Psi_N \rangle}{N} - \frac{\langle \Psi_N, (H_N-S_{v,\ell}) \Psi_N \rangle}{N} \right) \\
& \le \limsup_{N\to \infty} \left( \frac{\inf\sigma(H_N)}{N}-\frac{\inf\sigma(H_N-S_{v,\ell})}{N} \right) \\
&\le e_{\rm GP} - \inf_{\|u\|_{L^2}=1} \left( \cE_{\rm GP}(u) - |\langle v, u\rangle|^{2\ell}  \right).
\end{align*}
Here $v$ is not necessarily normalized. Therefore, we can replace $v$ by $\lambda^{1/(2\ell)} v$ with $\lambda>0$ and obtain
\begin{align} \label{eq:HF-1}
\limsup_{N\to \infty}\Tr \Big(|v^{\otimes \ell}\rangle \langle v^{\otimes \ell}| \gamma_{\Psi_N}^{(\ell)} \Big) \le \frac{1}{\lambda}\left( e_{\rm GP} - \inf_{\|u\|_{L^2}=1} \left( \cE_{\rm GP}(u) - \lambda |\langle v, u\rangle|^{2\ell}  \right)\right).
\end{align}
With given $v$ and $\ell$, for every $\lambda>0$ let $u_\lambda$ be a (normalized) minimizer for $u\mapsto \cE_{\rm GP}(u) - \lambda |\langle v, u\rangle|^{2\ell}$. Since $\langle u_\lambda, h u_\lambda \rangle$ is bounded and $h$ has compact resolvent, there exists a subsequence $\lambda_j\to 0$ such that $u_{\lambda_j}$ converges to $u_0$ in $L^2$. By Fatou's lemma, $u_0$ is a minimizer of $\cE_{\rm GP}(u)$. Moreover,
\begin{align} \label{eq:HF-2}
& \limsup_{j\to \infty }\frac{1}{\lambda_j}\left( e_{\rm GP} - \inf_{\|u\|_{L^2}=1} \left( \cE_{\rm GP}(u) - \lambda_j |\langle v, u\rangle|^{2\ell}  \right)\right) \nn\\
& \le \limsup_{j\to \infty }\frac{1}{\lambda_j} \left( \cE_{\rm GP}(u_{\lambda_j}) - \left( \cE_{\rm GP}(u_{\lambda_j})- \lambda_j |\langle v, u_{\lambda_j}\rangle|^{2\ell}  \right)\right) = |\langle v, u_0\rangle|^{2\ell} .
\end{align}
From \eqref{eq:HF-1} and \eqref{eq:HF-2}, we conclude that for every $v\in L^2(\R^3)$ and $\ell\in \mathbb{N}$,
\begin{align} \label{eq:HF-3}
\limsup_{N\to \infty}\Tr \Big(|v^{\otimes \ell}\rangle \langle v^{\otimes \ell}| \gamma_{\Psi_N}^{(\ell)} \Big) \le \sup_{u\in \cM_{\rm GP}} |\langle v, u\rangle|^{2\ell} 
\end{align}
where $\cM_{\rm GP}$ is the set of minimizers of $\cE_{\rm GP}(u)$. 

Note that also in \cite{LieSei-06} the upper bound \eqref{eq:HF-3} with $\ell=1$ was proved,  and from it the convergence of the one-particle density matrices was deduced using an abstract argument of convex analysis. In the following, we will provide a simpler way to conclude the convergence of density matrices from \eqref{eq:HF-3}, using the quantum de Finetti Theorem \ref{thm:DeFinetti}. Indeed, by Theorem \ref{thm:DeFinetti} as before, up to a subsequence of $\Psi_N$, there exists a Borel probability measure $\mu$ on the unit sphere $S\gH$ such that
\begin{align} \label{eq:HF-4 pre}
\lim_{N\to \infty} \Tr \left| \gamma_{\Psi_N}^{(k)} - \int |u^{\otimes k}\rangle \langle u^{\otimes k}| d\mu(u)\right| =0,\quad \forall k\in \mathbb{N}.
\end{align}
We will show that $\mu$ is supported on $\cM_{\rm GP}$. From \eqref{eq:HF-3} and \eqref{eq:HF-4 pre}, we get
\begin{align} \label{eq:HF-4}
\int |\langle  v,u\rangle|^{2k} d\mu(u) \le \sup_{u\in \cM_{\rm GP}} |\langle v, u\rangle|^{2k},\quad \forall v\in L^2(\R^3), k\in \mathbb{N}.
\end{align}

We assume for contradiction that there exists $v_0$ in the support of $\mu$ and $v_0\notin \cM_{\rm GP}$. We claim that we could then find $\delta\in (0,1/2)$ such that 
\begin{align} 
\label{eq:HF-5}|\langle v,u \rangle| \le 1-3\delta^2, \quad \forall u\in \cM_{\rm GP}, \forall v\in B
\end{align}
where $B$ is the set of all points in the support of $\mu$ within a $L^2$-distance less than $\delta$ from $v_0$. Indeed, if that were not the case, we would have two sequences strongly converging in $L ^2$ 
$$ v_n \to v_0, \quad u_n \to u_0 \in \cM_{\rm GP}$$
with $\norm{u_n - v_n} \to 0$, and thus $v_0 \in \cM_{\rm GP}$. Here we have used that $\cM_{\rm GP}$ is a compact subset of $L^2(\R^3)$.

On the other hand, by the triangle inequality,
\begin{align}
\label{eq:HF-6}|\langle v,u \rangle| \ge \frac{\|u\|^2+\|v\|^2-\|u-v\|^2}{2} \ge 1- 2\delta^2, \quad \forall u,v \in B.
\end{align}
Combining \eqref{eq:HF-4}, \eqref{eq:HF-5} and \eqref{eq:HF-6} we find that
\begin{align} \label{eq:HF-7}
(\mu(B))^2 (1-2\delta^2)^{2k} &\le \int_B\int_B |\langle  v,u\rangle|^{2k} d\mu(u) d\mu(v) \nn\\
&\le \int_B \sup_{u\in \cM_{\rm GP}} |\langle v, u\rangle|^{2k} d\mu(v) \le \mu(B) (1-3\delta^2)^{2k} 
\end{align}
for all $k\in \mathbb{N}$ and hence, taking $k\to \infty$,  $\mu(B)=0$. However, it is a contradiction to the fact that $v_0$ belongs to the support of $\mu$ and $\mu$ is a Borel measure. Thus we conclude that $\mu$ is supported on $\cM_{\rm GP}$ and the proof is complete. 
\end{proof}

\bibliographystyle{siam}

\end{document}